\documentclass[runningheads]{llncs}

\usepackage[T1]{fontenc}
\usepackage{graphicx}
\usepackage{amsmath}
\usepackage{amssymb}
\usepackage{mathtools}
\usepackage{algorithm}
\usepackage{algpseudocode}
\spnewtheorem{observation}[theorem]{Observation}{\bfseries}{\itshape}
\providecommand{\doi}[1]{} % keep DOI metadata in the .bib without printing DOI lines

\newcommand{\eps}{\varepsilon}
\newcommand{\pref}{\operatorname{pref}}
\newcommand{\suf}{\operatorname{suf}}

\newcommand{\TD}{TD_{\mathrm{NUNC}}}
\newcommand{\TDone}{TD^1_{\mathrm{NUNC}}}

\begin{document}

\title{An Approximation Algorithm for Non-uniform Non-contiguous Translocation Distance}
\titlerunning{An Approximation Algorithm for Non-uniform Non-contiguous Translocation Distance}

\author{Maria Constantin\inst{1} \and Adrian Micl\u{a}u\c{s}\inst{1} \and Alexandru Popa\inst{1, 2}}
\authorrunning{Constantin et al.}
\institute{Department of Computer Science, University of Bucharest, Str. Academiei 14, Bucharest, 010014, Romania \and
National Institute for Research and Development in Informatics, Bulevardul Mareșal Alexandru Averescu 8-10, Bucharest, 011555, Romania}
\maketitle

\begin{abstract}
Translocations are genome rearrangement operations that exchange prefixes of two chromosomes. We study the non-uniform non-contiguous translocation distance problem, where every string produced during the computation remains available for reuse. Given an initial set of strings $A$ and a target set $B$, the objective is to produce all strings in $B$ using as few translocations as possible.

We present the first polynomial-time approximation algorithm for this problem. For a single target string of length $n$, we obtain an $O(\log n)$-approximation, and we extend the result to arbitrary finite target sets with an $O(\log N)$-approximation, where $N$ is the total length of the targets not already present in the initial set. This resolves the approximability question for the non-uniform non-contiguous case left open by Constantin and Popa (TCS 2025).
\end{abstract}


\section{Introduction}
\label{sec:introduction}

\subsection*{Motivation}

Genome rearrangements are large-scale mutations that modify chromosome organization and are widely used to model evolutionary distance.  Typical operations include inversions, transpositions, duplications, fusions, fissions, and translocations~\cite{Sankoff1992,FertinEtAl2009}.  Algorithmically, one seeks a shortest sequence of allowed rearrangements transforming one genome representation into another.

Classical translocation models represent chromosomes as signed or unsigned permutations of distinct markers~\cite{KececiogluRavi1995,Hannenhalli1996}.  String-based models instead allow repeated symbols and formulate the problem directly on sets of strings~\cite{MartinVideMitrana2004,ConstantinPopa2019}.  A further distinction concerns availability: in the contiguous model, produced copies may be consumed, whereas in the non-contiguous model every produced string remains available and may be reused arbitrarily often.

A tight $2$-approximation is known for the non-uniform contiguous problem with one target~\cite{ConstantinPopa2025}, while NP-hardness was recently established for both non-uniform contiguous and non-contiguous variants~\cite{ConstantinMiclausPopa2026Hardness}.  The approximability of the non-contiguous case remained open.  We resolve this question by exploiting the connection between reusable translocation constructions and compressed string representations.

\subsection*{Informal Problem Definition}

Let $A$ be a finite set of initial strings. Given two available strings $x=x_1x_2$ and $y=y_1y_2$, a non-uniform translocation exchanges the prefixes $x_1$ and $y_1$ and produces $y_1x_2$ and $x_1y_2$. The split positions are arbitrary, and one exchanged prefix may be empty. In a non-contiguous sequence, every string in $A$ and every string produced earlier remains available forever. Given a finite target set $B$, the objective is to produce every string in $B$ using as few translocations as possible.

\subsection*{Previous and Related Work}

Genome rearrangement algorithms have been studied extensively in permutation models (see, e.g.,~\cite{Sankoff1992,FertinEtAl2009}).  For signed translocations, Hannenhalli~\cite{Hannenhalli1996} gave the first polynomial-time algorithm, followed by faster algorithms of Li et al.~\cite{LiQiWangZhu2004}, Wang et al.~\cite{WangZhuLiuMa2005}, and Bergeron, Mixtacki, and Stoye~\cite{BergeronMixtackiStoye2006}.  For unsigned translocations, Kececioglu and Ravi~\cite{KececiogluRavi1995} gave a $2$-approximation, Zhu and Wang~\cite{ZhuWang2006} proved NP-hardness and approximability: $1.75$~\cite{CuiWangZhu2007}, $1.5+\eps$~\cite{CuiWangZhuLiu2008}, $1.408+\eps$~\cite{JiangWangZhuZhu2014}, and $1.375$~\cite{PuZhuJiang2020}.

These classical models differ fundamentally from the string setting.  Permutation models use distinct markers, preserve a current genome configuration, and replace the two input chromosomes after a translocation~\cite{KececiogluRavi1995,Hannenhalli1996}.  In the string model, symbols may repeat within and across chromosomes, the input is a set of strings, and the task is generative: produce a target set from the initial set~\cite{MartinVideMitrana2004,ConstantinPopa2019}.  Most importantly, the non-contiguous model permits unrestricted reuse of every intermediate string~\cite{ConstantinMiclausPopa2026Hardness} . A single intermediate may therefore supply different prefixes or suffixes to many later operations.  Results for permutation translocations and consumption-based contiguous constructions do not directly transfer.

Translocation distance on strings was introduced in~\cite{MartinVideMitrana2004}.  Further exact cases were studied in~\cite{ConstantinPopa2019}, and a tight $2$-approximation for the non-uniform contiguous single-target case was obtained in~\cite{ConstantinPopa2025}.  NP-hardness for both non-uniform contiguous and non-contiguous variants, together with an FPT algorithm parameterized by the target length for one target, was recently shown in~\cite{ConstantinMiclausPopa2026Hardness}.

Our approximation analysis also draws on compressed string representations. Lempel--Ziv parsing~\cite{ZivLempel1977}, macro schemes~\cite{StorerSzymanski1982}, and relative Lempel--Ziv~\cite{KuruppuPuglisiZobel2010} represent repeated substrings by references. Collage systems allow concatenation, truncation, and repetition~\cite{KidaEtAl2003}, while straight-line programs (SLPs) and the smallest-grammar problem use concatenation rules~\cite{CharikarEtAl2005}. Balanced grammar constructions and logarithmic approximations for grammar size were developed in~\cite{Rytter2003,Jez2016}. Related notions of reuse are captured by bidirectional and ordered parsings~\cite{NavarroOchoaPrezza2021} and by string attractors~\cite{KempaPrezza2018}. Most directly, Migita, Uehata, and I~\cite{MigitaUehataI2026} showed that collage systems can be transformed into internal collage systems with only a constant-factor increase in size. We specialize their analysis to the repetition-free composition systems arising from translocation constructions.

\subsection{Our Results}

Our main result is the first polynomial-time approximation algorithm for the non-uniform non-contiguous translocation distance. For a single target string $z$ of length $n$, we obtain an $O(\log n)$-approximation. We further extend the result to arbitrary finite target sets, obtaining an $O(\log N)$-approximation, where $N$ is the total length of the target strings not already present in the initial set.

For the single-target case, the algorithm scans $z$ from left to right and greedily partitions it into factors. At each position, it chooses a longest factor that either occurs as a substring of an initial string or has an earlier non-overlapping occurrence in the already processed prefix of $z$. Each such factor can be extracted from an available string and appended to the constructed prefix using a constant number of 1-translocations.

The approximation analysis relates optimal non-contiguous translocation sequences to compact representations of the target. Starting from an optimal 1-translocation construction, we obtain a system in which strings are represented using concatenations and substrings of previously available strings. We internalize this representation using~\cite{MigitaUehataI2026} and then convert it, with logarithmic overhead, into a straight-line program (SLP), that is, an acyclic grammar that builds the target using binary concatenations. From this grammar we derive a short factorization of the type used by the algorithm, and a stay-ahead argument shows that the greedy factorization uses no more factors.

Combining these results gives an $O(\log n)$-approximation for one target and, by reducing to a concatenated target, an $O(\log N)$-approximation for arbitrary finite target sets. This resolves the approximability question for the non-uniform non-contiguous model left open in~\cite{ConstantinPopa2025}.

\section{Preliminaries}
\label{sec:preliminaries}

Let $\Sigma$ be a finite alphabet. We write $\Sigma^*$ for the set of all finite strings over $\Sigma$, $\Sigma^+=\Sigma^*\setminus\{\epsilon\}$, and $|x|$ for the length of a string $x$. For $1\le i\le j\le |x|$, $x[i..j]$ denotes the substring of $x$ starting at position $i$ and ending at position $j$. For $0\le p\le |x|$, we denote by $\pref_p(x)$ the prefix of $x$ of length $p$, and by $\suf_p(x)$ the suffix of $x$ of length $p$. In particular, $\pref_0(x)=\suf_0(x)=\epsilon$.

\begin{definition}[Translocation]
Let $x,y\in\Sigma^+$ and let $i,j$ be integers such that $0\le i\le |x|$ and $0\le j\le |y|$. Write $x=x_1x_2$ and $y=y_1y_2$, where $|x_1|=i$ and $|y_1|=j$. A \emph{non-uniform translocation} exchanges the prefixes $x_1$ and $y_1$ and produces the strings $u=y_1x_2$ and $v=x_1y_2$. We write
$(x,y)\vdash_{i,j}(u,v)$.
\end{definition}

\paragraph{Example.}
Let $x=\mathtt{abcd}$ and $y=\mathtt{efgh}$. Splitting $x$ after its second symbol and $y$ after its first symbol gives $x_1=\mathtt{ab}$, $x_2=\mathtt{cd}$, $y_1=\mathtt{e}$, and $y_2=\mathtt{fgh}$. Exchanging the two prefixes produces $\mathtt{ecd}$ and $\mathtt{abfgh}$. Thus $(\mathtt{abcd},\mathtt{efgh})\vdash_{2,1}(\mathtt{ecd},\mathtt{abfgh})$.

Thus, the two exchanged prefixes may have arbitrary and possibly different lengths. In particular, either prefix may be empty.

Let $A\subseteq\Sigma^+$ be a finite set of initial strings. In the non-contiguous model, every string in $A$ is available throughout the computation, and every string produced by a translocation remains available after its production and may be reused arbitrarily many times.

A non-contiguous translocation sequence is \emph{$B$-producing}, for a finite set $B\subseteq\Sigma^+$, if every string $z\in B$ either belongs to $A$ or is produced by at least one translocation in the sequence. The \emph{non-uniform non-contiguous translocation distance} $TD_{NUNC}(A,B)$ is the minimum length of a non-contiguous $B$-producing translocation sequence from $A$.

For a set of strings $S$, let $\operatorname{alph}(S)$ denote the set of symbols occurring in its strings. We assume throughout that $\operatorname{alph}(B)\subseteq\operatorname{alph}(A)$. This is exactly the feasibility condition: translocations cannot introduce a new symbol, while under this inclusion each required symbol can be extracted from an initial string, reused arbitrarily often, and concatenated to generate every string in $B$. Hence $TD_{NUNC}(A,B)<\infty$ under our standing assumption.

For the approximation analysis, we use a variant of the operation in which only one of the two outputs is retained.

\begin{definition}[1-translocation]
Let $(x,y)\vdash_{i,j}(u,v)$ be a non-uniform translocation. A \emph{1-translocation} performs the same exchange but retains only one of the two outputs. We write $(x,y)\vdash^1_{i,j}u$ when $u$ is retained and $(x,y)\vdash^1_{i,j}v$ when $v$ is retained.

The minimum length of a non-contiguous $B$-producing sequence of 1-translocations from $A$ is denoted by $TD^1_{NUNC}(A,B)$.
\end{definition}

\paragraph{Example.}
Let $x=\mathtt{abcd}$ and $y=\mathtt{efgh}$. Splitting $x$ after its second symbol and $y$ after its first symbol produces the outputs $\mathtt{ecd}$ and $\mathtt{abfgh}$. A 1-translocation retains only one of them; for example, $(\mathtt{abcd},\mathtt{efgh})\vdash^1_{2,1}\mathtt{abfgh}$.

After possibly exchanging the roles of the two input strings, every output retained by a 1-translocation can be written as $\pref_p(x)\suf_q(y)$, for some $0\le p\le |x|$ and $0\le q\le |y|$.

The two distance measures differ by at most a factor of two.

\begin{lemma}
\label{lem:one-vs-two}
For all finite $A,B\subseteq\Sigma^+$,
$TD_{NUNC}(A,B)\le TD^1_{NUNC}(A,B)\le 2\,TD_{NUNC}(A,B)$.
\end{lemma}

\begin{proof}
Let $S=(s_1,\ldots,s_{TD_{NUNC}(A,B)})$ be an optimal non-contiguous $B$-producing translocation sequence, where
$s_i=(x_i,y_i)\vdash_{p_i,q_i}(u_i,v_i)$.

First, every 1-translocation can be executed as the corresponding ordinary translocation while ignoring the output that is not retained. Hence every non-contiguous $B$-producing sequence of 1-translocations yields a non-contiguous $B$-producing sequence of ordinary translocations of the same length. Therefore,
$TD_{NUNC}(A,B)\le TD^1_{NUNC}(A,B)$.

For the converse inequality, replace every translocation
$s_i=(x_i,y_i)\vdash_{p_i,q_i}(u_i,v_i)$
by the two 1-translocations
$(x_i,y_i)\vdash^1_{p_i,q_i}u_i$
and
$(x_i,y_i)\vdash^1_{p_i,q_i}v_i$.
Because the model is non-contiguous, using $x_i$ and $y_i$ in the first 1-translocation does not consume them, and the same inputs remain available for the second one. Applying this replacement to every operation of $S$ yields a valid non-contiguous $B$-producing sequence of $2\,TD_{NUNC}(A,B)$ 1-translocations. Consequently,
$TD^1_{NUNC}(A,B)\le 2\,TD_{NUNC}(A,B)$. \qed
\end{proof}

We next introduce the compressed representations used only in the approximation analysis.

\begin{definition}[$A$-relative composition system]
An \emph{$A$-relative composition system} is an acyclic collection of non-free variables, each defined by one of the following two types of rules:
$X\to YZ$ or $X\to Y[i..j]$.
Here, the right-hand side may refer either to previously defined non-free variables or to substrings of strings in $A$. Every substring of a string in $A$ is regarded as a free source and is not counted in the size of the system.

Acyclicity means that the non-free variables admit an ordering $X_1,\ldots,X_m$ such that the rule defining $X_r$ refers only to free sources and to variables $X_s$ with $s<r$. A distinguished start variable evaluates to the target string. The size of the system is the number of non-free variables, equivalently the number of non-free rules.
\end{definition}

\begin{definition}[Internal $A$-relative composition system]
An $A$-relative composition system generating a string $z$ is \emph{internal} if every non-free variable occurs untruncated in the derivation of the start variable. Consequently, the complete value of every non-free variable occurs as a substring of $z$, and hence every such variable has length at most $|z|$.
\end{definition}

\begin{definition}[$A$-relative SLP]
An \emph{$A$-relative straight-line program}, or \emph{$A$-SLP}, is an acyclic binary grammar in which every non-free rule has the form $X\to YZ$, while arbitrary substrings of strings in $A$ may be used as free sources. The size of an $A$-SLP is the number of its non-free variables. We denote by $g_A(z)$ the minimum size of an $A$-SLP generating $z$.
\end{definition}

Finally, we introduce the factorization used by our approximation algorithm.

\begin{definition}[$A$-relative non-overlapping factorization]
\label{def:relative-factorization}
Let $z\in\Sigma^+$ and let $z=F_1F_2\cdots F_q$ be a factorization. For every $i\in\{1,\ldots,q\}$, let
$s_i=1+\displaystyle\sum_{h<i}|F_h|$
be the starting position of $F_i$ in $z$.

The factorization is an \emph{$A$-relative non-overlapping factorization} if every factor $F_i$ satisfies at least one of the following conditions:
\begin{enumerate}
\item $F_i$ occurs as a substring of some string in $A$  or
\item there exists a position $p<s_i$ such that
$F_i=z[p..p+|F_i|-1]$
and
$p+|F_i|-1<s_i$.
\end{enumerate}
\end{definition}

In the second case, the earlier occurrence of $F_i$ lies completely before the current occurrence. Hence, when the factors are constructed from left to right, the source of $F_i$ is already contained in the previously constructed prefix of $z$.

\section{The Single-Target Approximation}
\label{sec:single}

Throughout this section $B=\{z\}$, $n=|z|$, and $k=\TDone(A,\{z\})$.  We first describe the greedy factorization algorithm and its implementation, then show how to realize the resulting factors by translocations, and finally bound the number of factors by comparing the algorithm with an optimal reusable construction.

\subsection{The Greedy Factorization Algorithm}
\label{subsec:algorithm}

At a position $s$, let $L_A[s]$ be the length of the longest prefix of $z[s..n]$ occurring as a substring of some string in $A$, and let $L_N[s]$ be the length of the longest prefix having an earlier occurrence ending strictly before $s$. The latter is the \emph{Longest Previous non-overlapping Factor} (LPnF) value at $s$~\cite{CrochemoreTischler2011}. The algorithm chooses a factor of length $\max\{L_A[s],L_N[s]\}$, stores a witnessing source, and advances past the factor (ties are arbitrary).

\begin{algorithm}[t]
\caption{Greedy $A$-relative non-overlapping factorization}
\label{alg:greedy-factorization}
\begin{algorithmic}[1]
\Require Initial set $A$ and target $z=z[1..n]$
\State compute $(L_A[s],\mathrm{src}_A[s])$ and $(L_N[s],\mathrm{src}_N[s])$ for all $s\in\{1,\ldots,n\}$
\State $s\gets1$; $\mathcal F\gets()$
\While{$s\le n$}
    \If{$L_A[s]\ge L_N[s]$}
        \State $\ell\gets L_A[s]$; $\mathrm{src}\gets\mathrm{src}_A[s]$
    \Else
        \State $\ell\gets L_N[s]$; $\mathrm{src}\gets\mathrm{src}_N[s]$
    \EndIf
    \State append $(z[s..s+\ell-1],\mathrm{src})$ to $\mathcal F$
    \State $s\gets s+\ell$
\EndWhile
\State \Return $\mathcal F$
\end{algorithmic}
\end{algorithm}

\begin{lemma}
\label{lem:greedy-correctness-time}
Algorithm~\ref{alg:greedy-factorization} computes the greedy $A$-relative non-overlapping factorization of $z$ in $O(\mathcal L)$ time and space on an integer alphabet, where $\mathcal L=n+\displaystyle\sum_{a\in A}|a|+|A|+1$.
\end{lemma}

\begin{proof}
For correctness, every valid factor starting at $s$ either occurs in an initial string or has an earlier non-overlapping occurrence, so its length is at most $L_A[s]$ or $L_N[s]$, respectively. Conversely, the definitions of these values provide witnessing occurrences, hence a factor of length $\max\{L_A[s],L_N[s]\}$ is valid. The algorithm therefore chooses a longest valid factor at every boundary and outputs exactly the greedy factorization. By the standing alphabet assumption, $L_A[s]\ge1$ for every $s$, so it always advances.

For the running time, write $A=\{a_1,\ldots,a_m\}$. Form $W=z\#_0a_1\#_1\allowbreak\cdots a_m\#_m$ using pairwise distinct separators outside the input alphabet. Build the suffix array and LCP array of $W$~\cite{KarkkainenSanders2003,KasaiEtAl2001} and mark suffixes starting inside strings of $A$. For a suffix starting at $z[s]$, the maximum LCP with a marked suffix is attained by the nearest marked suffix on either side in suffix-array order: a farther suffix uses a larger LCP interval and cannot yield a larger minimum. Two scans find these nearest marked suffixes, and RMQ on the LCP array gives the corresponding values and witnesses in constant time, yielding all $L_A[s]$ in $O(\mathcal L)$ time and space after suffix-array construction. The values $L_N[s]$ are the LPnF table of $z$, computable with witnesses in $O(n)$ time and space from the suffix array~\cite{CrochemoreTischler2011}. After preprocessing, the algorithm spends constant time per factor, hence $O(n)$ in total. Linear-time suffix-array construction~\cite{KarkkainenSanders2003} gives $O(\mathcal L)$ total time and space on an integer alphabet. \qed
\end{proof}

\paragraph{Example.}
For $A=\{a,b,c\}$ and $z=\mathtt{abcabcabc}$, Algorithm~\ref{alg:greedy-factorization} returns $a\mid b\mid c\mid\mathtt{abc}\mid\mathtt{abc}$. The first three factors use sources in $A$. At the fourth, $L_N[4]=3$ since $z[1..3]=z[4..6]$, and after appending it the final $\mathtt{abc}$ again has an earlier non-overlapping source.

\subsection{Realizing the Factorization by Translocations}
\label{subsec:realization}

Let the factorization returned by Algorithm~\ref{alg:greedy-factorization} be $z=F_1\cdots F_q$, and let $P_i=F_1\cdots F_i$, with $P_0=\eps$.  When processing $F_i$, its stored source is either a substring of some $a\in A$ or an earlier non-overlapping occurrence contained completely in $P_{i-1}$. Hence the source string is already available.

\begin{lemma}
\label{lem:extract}
If a string $w$ is available, every non-empty substring $w[i..j]$ can be made available using at most two 1-translocations.
\end{lemma}

\begin{proof}
Choose any helper $a\in A$.  First split $w$ immediately before position $i$ and split $a$ at position $0$. Retaining the output consisting only of the suffix of $w$ produces $w[i..|w|]$.  From this suffix, retain its prefix of length $j-i+1$ by a second translocation in which the other output contribution is empty.  The result is $w[i..j]$. \qed
\end{proof}

\begin{lemma}
\label{lem:append}
If $x$ and $y$ are available, then $xy$ can be produced by one 1-translocation.
\end{lemma}

\begin{proof}
Split $x$ after its last symbol and $y$ before its first symbol. The resulting 1-translocation retains the output $xy$. \qed
\end{proof}

For the first factor, one can simply extract $F_1$ from its $A$-source and designate it as $P_1$. For every later factor, after extracting $F_i$ we append it to $P_{i-1}$.  We use the convenient uniform bound of three operations per factor.

\begin{lemma}
\label{lem:factorization-realization}
If $z$ has an $A$-relative non-overlapping factorization with $q$ factors, then $z$ can be produced using at most $3q$ 1-translocations.
\end{lemma}

\begin{proof}
When $F_i$ is processed, its source is an available string: either an initial string or the already constructed prefix $P_{i-1}$.  Lemma~\ref{lem:extract} makes $F_i$ available in at most two operations, and Lemma~\ref{lem:append} appends it in one more operation.  Summing over all factors gives at most $3q$ operations. \qed
\end{proof}

Thus it remains only to bound the greedy factor count.

\subsection{From an Optimal Translocation DAG to a Relative Composition System}
\label{subsec:dag-composition}

We first make explicit why an optimal sequence can be viewed as a compact reusable DAG.

\begin{lemma}[Construction-DAG normalization]
\label{lem:dag-normalization}
There exists an optimal 1-translocation sequence for $z$ in which no generated string is produced twice and every generated string has a dependency path to the target.  Hence an optimum of length $k$ is represented by an acyclic construction DAG with exactly $k$ generated vertices.
\end{lemma}

\begin{proof}
Suppose a string $X$ is produced twice.  Delete the later production and redirect every later use of that copy to the first occurrence of $X$.  This does not change any available input because, in the non-contiguous model, the first $X$ remains available forever.  Repeating the argument removes all duplicate productions.  Next remove any operation whose retained output has no dependency path to $z$. Such an output is irrelevant to the production of the target.  Ordering the remaining generated vertices by their execution time is a topological order, so the dependency graph is acyclic.  There is one generated vertex for each of the $k$ 1-translocations. \qed
\end{proof}

Fix such a normalized optimum.  Each generated string \(X\) has a production \(X=\pref_p(Y)\suf_q(Z)\), where \(Y,Z\) are either initial strings or earlier generated strings and one of the two contributions may be empty.

\begin{lemma}
\label{lem:dag-composition}
A 1-translocation construction DAG with \(k\) generated vertices can be converted into an equivalent \(A\)-relative composition system with \(c\) concatenation rules and \(t\) truncation rules satisfying \(c\le k\), \(t\le k+c\), and hence \(c+t\le3k\).
\end{lemma}

\begin{proof}
Consider one generated vertex \(X=\pref_p(Y)\suf_q(Z)\).  If both contributions are non-empty, introduce the necessary prefix and suffix variables, \(P_X\to Y[1..p]\) and \(S_X\to Z[|Z|-q+1..|Z|]\), whenever those pieces are not already free or equal to the full source, and then use \(X\to P_XS_X\).  Thus every two-sided 1-translocation contributes exactly one concatenation rule and at most two truncation rules.  If one contribution is empty, \(X\) is only a prefix or suffix of one source and needs at most one truncation rule.

Let \(c\) be the number of two-sided operations.  Then \(c\le k\).  Those \(c\) operations contribute at most \(2c\) truncations, while the other \(k-c\) operations contribute at most one each, so \(t\le2c+(k-c)=k+c\).  Consequently \(c+t\le k+2c\le3k\).

The helper variables are inserted immediately before \(X\) in the topological order of the original DAG.  Hence every new rule refers only to earlier variables or free \(A\)-substrings, acyclicity is preserved, and an induction over the topological order shows that every original DAG vertex evaluates to the same string in the composition system. \qed
\end{proof}

\subsection{Internalization}
\label{subsec:internalization}

The system from Lemma~\ref{lem:dag-composition} need not be internal: it may build a long intermediate variable and later use only a truncation of it.  Migita, Uehata, and I~\cite{MigitaUehataI2026} eliminate exactly this phenomenon.  They call a nonterminal reachable if it appears untruncated in the grammar tree. An internal collage system is one in which every nonterminal is reachable.  Their Theorem~3 converts a collage system of size \(m\) into an internal one in \(O(m^2)\) time, and their size analysis gives the explicit published bound \(9m-4m_{\rm tr}\le9m\), where \(m_{\rm tr}\) is the number of original truncation rules~\cite{MigitaUehataI2026}.  Since composition systems are collage systems without repetition rules~\cite{KidaEtAl2003,MigitaUehataI2026}, their case analysis admits a sharper bound in our setting.

\begin{lemma}[Relative internalization]
\label{lem:relative-internalization}
Let an \(A\)-relative composition system have \(c\) concatenation rules and \(t\) truncation rules, \(m=c+t\), and generate \(z\).  It can be transformed in polynomial time into an equivalent internal \(A\)-relative composition system of size at most \(5m-2t=5c+3t\).
\end{lemma}

\begin{proof}
We follow the internalization algorithm of Migita et al.~\cite{MigitaUehataI2026}.  Process unreachable variables top-down.  Such a variable \(X\) is referenced only through truncation rules: if a reachable concatenation used \(X\) directly, then \(X\) would itself be reachable.  When \(X\to YZ\) and a truncation \(Q\to X[i..j]\) lies wholly inside \(Y\) or wholly inside \(Z\), their Case~1 simply redirects the truncation to the appropriate child, without increasing the number of rules.  If the requested interval crosses the boundary \(Y|Z\), their Case~3 processes all crossing truncations together.  Let \(s\) be the longest suffix requested from \(Y\) and \(p\) the longest prefix requested from \(Z\).  The construction introduces one representative \(S=\suf_s(Y)\) and one representative \(P=\pref_p(Z)\). Every crossing request is then represented by a suffix of \(S\), a prefix of \(P\), and one concatenation.  The net size increase is at most two per truncation removed in this case~\cite{MigitaUehataI2026}.

The only case in their proof with a size increase of four is Case~2, which handles a truncation crossing more than two copies of a repetition rule \(X\to Y^r\)~\cite{MigitaUehataI2026}.  Our composition system contains no repetition rules, so Case~2 never occurs.  The remaining bookkeeping is exactly the one in their size analysis: besides the \(t\) original truncation rules, Case~3 can create only the two representative truncations \(S\) and \(P\) for each original non-truncation rule, hence at most \(2(m-t)\) additional truncations can later require processing~\cite{MigitaUehataI2026}.  Charging the repetition-free worst-case increase \(2\), rather than the general collage-system charge \(4\), gives \(m+2(t+2(m-t))=5m-2t\).

The same argument applies in the relative setting.  If downward redirection reaches an initial string \(a\in A\), the requested substring of \(a\) is used directly as a free source and creates no counted variable.  At termination every non-free variable is reachable, hence appears untruncated in the derivation of \(z\) and therefore evaluates to a substring of \(z\). \qed
\end{proof}

\begin{observation}
\label{obs:direct-cpm}
If one prefers to use only the theorem stated verbatim in~\cite{MigitaUehataI2026}, apply the published bound \(9m\) directly.  Lemma~\ref{lem:dag-composition} gives \(m\le3k\), so this yields an internal system of size at most \(27k\).  The asymptotic approximation ratio is unchanged.
\end{observation}

For the sharper accounting, Lemma~\ref{lem:dag-composition} gives \(c\le k\) and \(t\le k+c\).  Lemma~\ref{lem:relative-internalization} therefore yields \(5c+3t\le5c+3(k+c)=3k+8c\le11k\).  Thus an optimal 1-translocation solution of size \(k\) induces an internal relative composition system of size at most \(11k\).  In particular every non-free variable represents a substring of \(z\) and has length at most \(n\).

\subsection{From an Internal Composition System to an SLP}
\label{subsec:composition-slp}

We now eliminate substring rules at logarithmic cost.  The construction is a persistent balanced version of the classical grammar-balancing viewpoint used in grammar compression~\cite{Rytter2003,CharikarEtAl2005,Jez2016}.

\begin{lemma}
\label{lem:composition-slp}
If an internal $A$-relative composition system of size $h$ generates a string $z$ of length $n$, then one can construct in polynomial time an $A$-SLP for $z$ with $O(h\log(n+1))$ non-free variables.
\end{lemma}

\begin{proof}
Process the composition-system variables in topological order. For every processed variable $X$, let $|X|$ denote the length of the string represented by $X$, and maintain a persistent balanced binary concatenation tree $T_X$ whose yield is this string. Leaves are free substrings of strings in $A$, and every internal tree node is an ordinary binary SLP rule. We maintain height $O(\log(|X|+1))$.

Two standard persistent operations suffice.  $\operatorname{Join}(T_1,T_2)$ returns a balanced tree for the concatenation of the two yields, while $\operatorname{Split}(T,p)$ returns balanced trees for the prefix of length $p$ and the complementary suffix.  Store at every internal node the length of the string represented by its left subtree.  A split follows one root-to-leaf search path determined by these lengths. Subtrees not on this path are shared with the old version, and rebuilding plus rebalancing creates only $O(1)$ new nodes per level.  A join descends along the appropriate spine of the taller tree until the two heights are compatible and rebalances on the return path.  Since every maintained tree has logarithmic height, either operation creates $O(\log(n+1))$ new nodes.  This is the standard balancing principle behind AVL-grammar constructions~\cite{Rytter2003}.  If a split point falls inside a free $A$-substring leaf, that leaf is replaced by the two corresponding substrings of the same initial string, which remain free.

For a concatenation rule $X\to YZ$, set $T_X=\operatorname{Join}(T_Y,T_Z)$.  For a substring rule $X\to Y[i..j]$, first split $T_Y$ after position $j$ and then split the resulting prefix after position $i-1$. The middle tree is $T_X$.  Thus either type of composition rule creates only $O(\log(n+1))$ non-free SLP variables.

Internality is used here: every non-free composition variable is a substring of $z$, hence has length at most $n$, so all maintained trees have logarithmic height in $n$.  Persistence preserves earlier trees, and processing in topological order preserves acyclicity.  Summing the logarithmic overhead over the $h$ rules proves the claim. \qed
\end{proof}

We obtain the following comparison bound.

\begin{corollary}
\label{cor:slp-vs-opt}
For \(B=\{z\}\) and \(n=|z|\), \(g_A(z)=O(\TDone(A,B)\log(n+1))\).
\end{corollary}

\subsection{From the SLP to a Short Factorization}
\label{subsec:slp-factorization}

The next lemma is the bridge from a reusable grammar to the left-to-right factorization used by the algorithm.

\begin{lemma}
\label{lem:slp-factorization}
For every target $z$, the greedy $A$-relative non-overlapping factorization has at most $g_A(z)+1$ factors.
\end{lemma}

\begin{proof}
Let $G$ be an $A$-SLP for $z$ with $g=g_A(z)$ non-free variables and unfold it into its complete ordered derivation tree.  Every node occurrence spells a contiguous interval of $z$.  For every non-free variable $X$, keep only its leftmost occurrence expanded. Prune every later occurrence of $X$ and regard that occurrence itself as a leaf.  Free $A$-substring leaves are never expanded.

Each non-free variable is expanded at most once, so the pruned binary tree has at most $g$ internal nodes and therefore at most $g+1$ leaves.  Reading these leaves from left to right yields a factorization $z=H_1\cdots H_r$ with $r\le g+1$.

Every leaf is a valid relative non-overlapping factor.  A free leaf is, by definition, a substring of a string in $A$.  Otherwise the leaf is a pruned occurrence of some variable $X$.  Its leftmost occurrence spells the same string and begins earlier.  The two occurrences cannot overlap: intervals of derivation-tree nodes are either disjoint or nested, while nesting two occurrences of the same variable would make one occurrence a proper descendant of the other and hence create a dependency path $X\Rightarrow^+X$, contradicting acyclicity.  Thus the leftmost occurrence ends strictly before the pruned one begins.

It remains to compare the greedy factorization with this factorization.  More generally, let $0=p_0<p_1<\cdots<p_r=n$ be the boundaries of any valid relative non-overlapping factorization and $0=g_0<g_1<\cdots$ the greedy boundaries.  We prove by induction that $g_i\ge p_i$.  Assume $g_{i-1}\ge p_{i-1}$.  If $g_{i-1}\ge p_i$ there is nothing to prove.  Otherwise the greedy starting position lies inside the comparison factor ending at $p_i$.  The suffix of that factor from the greedy start to $p_i$ is still valid: a suffix of an $A$-substring is an $A$-substring, and a suffix of a factor with an earlier non-overlapping source has the corresponding earlier suffix source.  Hence a valid greedy candidate reaches at least $p_i$, and the longest valid candidate reaches no less far.  Therefore $g_i\ge p_i$.  The greedy factorization reaches the end using at most $r\le g+1$ factors. \qed
\end{proof}

Combining the previous results gives the single-target guarantee.

\begin{theorem}
\label{thm:single-target}
For $B=\{z\}$, the greedy relative non-overlapping factorization algorithm constructs $z$ in polynomial time using $O\!\left(\TDone(A,\{z\})\log(n+1)\right)$ 1-translocations.  Consequently it is an $O(\log(n+1))$-approximation for $\TD(A,\{z\})$.
\end{theorem}

\begin{proof}
Let $k=\TDone(A,\{z\})$.  Corollary~\ref{cor:slp-vs-opt} and Lemma~\ref{lem:slp-factorization} imply that the greedy factorization has $q=O(k\log(n+1))$ factors.  Lemma~\ref{lem:factorization-realization} realizes it with at most $3q=O(k\log(n+1))$ 1-translocations.  Finally, Lemma~\ref{lem:one-vs-two} gives $k\le 2\TD(A,\{z\})$, and every 1-translocation can be executed as an ordinary translocation without increasing the number of operations. \qed
\end{proof}

\section{Several Target Strings}
\label{sec:multiple}

The single-target result extends with only constant additional loss.  Let $B'=B\setminus A=\{z_1,\ldots,z_r\}$ and let $N=\displaystyle\sum_{i=1}^r|z_i|$.  If $B'=\emptyset$, no operation is needed.  Otherwise let $m=\TDone(A,B)$.  Because one 1-translocation retains only one output, every target in $B'$ must be produced by some operation, and hence $r\le m$.  Define the concatenated target $T=z_1z_2\cdots z_r$.

\begin{lemma}
\label{lem:combined-target}
$\TDone(A,\{T\})\le 2m$.
\end{lemma}

\begin{proof}
Execute an optimal $B$-producing 1-translocation sequence of length $m$.  All $z_i$ are then available.  Concatenating them successively requires $r-1$ further 1-translocations, so $\TDone(A,\{T\})\le m+r-1\le 2m$. \qed
\end{proof}

Apply Theorem~\ref{thm:single-target} to $T$, whose length is $N$.  This constructs $T$ using $O(m\log(N+1))$ operations.  Afterwards every $z_i$ is a substring of the available string $T$ and can be extracted in at most two operations by Lemma~\ref{lem:extract}.  The extraction cost is at most $2r\le 2m$. 

\begin{theorem}
\label{thm:general-target}
For arbitrary finite target sets $B$, there is a polynomial-time algorithm returning a $B$-producing sequence of length $O\!\left(\TD(A,B)\log(N+1)\right)$, where $N=\displaystyle\sum_{z\in B\setminus A}|z|$.
\end{theorem}

\begin{proof}
Let $m=\TDone(A,B)$. The construction above uses $O(m\log(N+1))$ operations in the 1-translocation model. By Lemma~\ref{lem:one-vs-two}, $m\le 2\TD(A,B)$. Replacing each 1-translocation by the corresponding ordinary translocation therefore yields a $B$-producing sequence of length $O(\TD(A,B)\log(N+1))$. \qed

\end{proof}

\section{Conclusion}
\label{sec:conclusion}

We gave a polynomial-time $O(\log(N+1))$-approximation for non-uniform non-contiguous translocation distance, resolving the approximation question left open by the contiguous-case study of Constantin and Popa~\cite{ConstantinPopa2025}.  The algorithm computes a greedy relative non-overlapping factorization. The analysis passes from an optimal reusable translocation DAG through relative composition systems and internalization to an SLP and finally back to a short factorization.  The transformations before the SLP lose only constant factors, and the logarithmic factor appears in the substring-to-concatenation compilation. It remains open whether a constant-factor approximation is possible and what approximation lower bounds hold.

\renewcommand{\refname}{References}
\bibliographystyle{splncs04}
\bibliography{references}

@article{Sankoff1992,
  title={Gene order comparisons for phylogenetic inference: evolution of the mitochondrial genome.},
  author={Sankoff, David and Leduc, Guillame and Antoine, Natalie and Paquin, Bruno and Lang, B Franz and Cedergren, Robert},
  journal={Proceedings of the National Academy of Sciences},
  volume={89},
  number={14},
  pages={6575--6579},
  year={1992}
}

@book{FertinEtAl2009,
  title={Combinatorics of genome rearrangements},
  author={Fertin, Guillaume and Labarre, Anthony and Rusu, Irena and Tannier, Eric and Vialette, St{\~A} and others},
  year={2009},
  publisher={The MIT Press}
}

@inproceedings{KececiogluRavi1995,
  title={Of mice and men: Algorithms for evolutionary distances between genomes with translocation},
  author={Kececioglu, John D and Ravi, R},
  booktitle={Symposium on discrete algorithms},
  volume={604},
  pages={613},
  year={1995}
}

@article{Hannenhalli1996,
  title={Polynomial-time algorithm for computing translocation distance between genomes},
  author={Hannenhalli, Sridhar},
  journal={Discrete applied mathematics},
  volume={71},
  number={1-3},
  pages={137--151},
  year={1996},
  publisher={Elsevier}
}

@inproceedings{MartinVideMitrana2004,
  title={A new uniform translocation distance},
  author={Mart{\'\i}n-Vide, Carlos and Mitrana, Victor},
  booktitle={Joint IAPR International Workshops on Statistical Techniques in Pattern Recognition (SPR) and Structural and Syntactic Pattern Recognition (SSPR)},
  pages={278--286},
  year={2004},
  organization={Springer}
}

@article{ConstantinPopa2019,
  title={Some remarks on the translocation distance},
  author={Constantin, Maria and Popa, Alexandru},
  journal={Procedia Computer Science},
  volume={159},
  pages={1757--1766},
  year={2019},
  publisher={Elsevier}
}

@article{ConstantinPopa2025,
  title={Exact and approximation algorithms for the contiguous translocation distance problem},
  author={Constantin, Maria and Popa, Alexandru},
  journal={Theoretical Computer Science},
  volume={1026},
  pages={115003},
  year={2025},
  publisher={Elsevier}
}

@misc{ConstantinMiclausPopa2026Hardness,
  author        = {Constantin, Maria and Micl{\u a}u{\c s}, Adrian and Popa, Alexandru},
  title         = {{NP}-Hardness and a Fixed-Parameter Algorithm for Translocation Distance},
  year          = {2026},
  eprint        = {2609.18397},
  archivePrefix = {arXiv},
  primaryClass  = {cs.DS},
  note          = {arXiv:2609.18397}
}

@inproceedings{LiQiWangZhu2004,
  title={A linear-time algorithm for computing translocation distance between signed genomes},
  author={Li, Guojun and Qi, Xingqin and Wang, Xiaoli and Zhu, Binhai},
  booktitle={Annual Symposium on Combinatorial Pattern Matching},
  pages={323--332},
  year={2004},
  organization={Springer}
}

@article{WangZhuLiuMa2005,
  title={An O (n2) algorithm for signed translocation},
  author={Wang, Lusheng and Zhu, Daming and Liu, Xiaowen and Ma, Shaohan},
  journal={Journal of Computer and System Sciences},
  volume={70},
  number={3},
  pages={284--299},
  year={2005},
  publisher={Elsevier}
}

@article{BergeronMixtackiStoye2006,
  title={On sorting by translocations},
  author={Bergeron, Anne and Mixtacki, Julia and Stoye, Jens},
  journal={Journal of Computational Biology},
  volume={13},
  number={2},
  pages={567--578},
  year={2006},
  publisher={SAGE Publications Sage CA: Los Angeles, CA}
}

@article{ZhuWang2006,
  title={On the complexity of unsigned translocation distance},
  author={Zhu, Daming and Wang, Lusheng},
  journal={Theoretical computer science},
  volume={352},
  number={1-3},
  pages={322--328},
  year={2006},
  publisher={Elsevier}
}

@article{CuiWangZhu2007,
  title={A 1.75-approximation algorithm for unsigned translocation distance},
  author={Cui, Yun and Wang, Lusheng and Zhu, Daming},
  journal={Journal of Computer and System Sciences},
  volume={73},
  number={7},
  pages={1045--1059},
  year={2007},
  publisher={Elsevier}
}

@article{CuiWangZhuLiu2008,
  title={A (1.5+ $\varepsilon$)-approximation algorithm for unsigned translocation distance},
  author={Cui, Yun and Wang, Lusheng and Zhu, Daming and Liu, Xiaowen},
  journal={IEEE/ACM Transactions on Computational Biology and Bioinformatics},
  volume={5},
  number={1},
  pages={56--66},
  year={2008},
  publisher={IEEE}
}

@inproceedings{JiangWangZhuZhu2014,
  title={A (1.408+ $\varepsilon$)-approximation algorithm for sorting unsigned genomes by reciprocal translocations},
  author={Jiang, Haitao and Wang, Lusheng and Zhu, Binhai and Zhu, Daming},
  booktitle={International Workshop on Frontiers in Algorithmics},
  pages={128--140},
  year={2014},
  organization={Springer}
}

@article{PuZhuJiang2020,
  title={A 1.375-approximation algorithm for unsigned translocation sorting},
  author={Pu, Lianrong and Zhu, Daming and Jiang, Haitao},
  journal={Journal of Computer and System Sciences},
  volume={113},
  pages={163--178},
  year={2020},
  publisher={Elsevier}
}

@article{ZivLempel1977,
  title={A universal algorithm for sequential data compression},
  author={Ziv, Jacob and Lempel, Abraham},
  journal={IEEE Transactions on information theory},
  volume={23},
  number={3},
  pages={337--343},
  year={1977},
  publisher={IEEE}
}

@article{StorerSzymanski1982,
  title={Data compression via textual substitution},
  author={Storer, James A and Szymanski, Thomas G},
  journal={Journal of the ACM (JACM)},
  volume={29},
  number={4},
  pages={928--951},
  year={1982},
  publisher={ACM New York, NY, USA}
}

@inproceedings{KuruppuPuglisiZobel2010,
  title={Relative Lempel-Ziv compression of genomes for large-scale storage and retrieval},
  author={Kuruppu, Shanika and Puglisi, Simon J and Zobel, Justin},
  booktitle={International Symposium on String Processing and Information Retrieval},
  pages={201--206},
  year={2010},
  organization={Springer}
}

@article{KidaEtAl2003,
  title={Collage system: a unifying framework for compressed pattern matching},
  author={Kida, Takuya and Matsumoto, Tetsuya and Shibata, Yusuke and Takeda, Masayuki and Shinohara, Ayumi and Arikawa, Setsuo},
  journal={Theoretical Computer Science},
  volume={298},
  number={1},
  pages={253--272},
  year={2003},
  publisher={Elsevier}
}

@article{CharikarEtAl2005,
  title={The smallest grammar problem},
  author={Charikar, Moses and Lehman, Eric and Liu, Ding and Panigrahy, Rina and Prabhakaran, Manoj and Sahai, Amit and Shelat, Abhi},
  journal={IEEE Transactions on Information Theory},
  volume={51},
  number={7},
  pages={2554--2576},
  year={2005},
  publisher={IEEE}
}

@article{Rytter2003,
  title={Application of Lempel--Ziv factorization to the approximation of grammar-based compression},
  author={Rytter, Wojciech},
  journal={Theoretical Computer Science},
  volume={302},
  number={1-3},
  pages={211--222},
  year={2003},
  publisher={Elsevier}
}

@article{Jez2016,
  title={A really simple approximation of smallest grammar},
  author={Je{\.z}, Artur},
  journal={Theoretical Computer Science},
  volume={616},
  pages={141--150},
  year={2016},
  publisher={Elsevier}
}

@article{NavarroOchoaPrezza2021,
  title={On the approximation ratio of ordered parsings},
  author={Navarro, Gonzalo and Ochoa, Carlos and Prezza, Nicola},
  journal={IEEE Transactions on Information Theory},
  volume={67},
  number={2},
  pages={1008--1026},
  year={2020},
  publisher={IEEE}
}

@inproceedings{KempaPrezza2018,
  title={At the roots of dictionary compression: string attractors},
  author={Kempa, Dominik and Prezza, Nicola},
  booktitle={Proceedings of the 50th Annual ACM SIGACT Symposium on Theory of Computing},
  pages={827--840},
  year={2018}
}

@InProceedings{MigitaUehataI2026,
  author =	{Migita, Soichiro and Uehata, Kyotaro and I, Tomohiro},
  title =	{{On the Smallest Size of Internal Collage Systems}},
  booktitle =	{37th Annual Symposium on Combinatorial Pattern Matching (CPM 2026)},
  pages =	{31:1--31:14},
  series =	{Leibniz International Proceedings in Informatics (LIPIcs)},
  ISBN =	{978-3-95977-420-8},
  ISSN =	{1868-8969},
  year =	{2026},
  volume =	{369},
  editor =	{Bille, Philip and Prezza, Nicola},
  publisher =	{Schloss Dagstuhl -- Leibniz-Zentrum f{\"u}r Informatik},
  address =	{Dagstuhl, Germany}
}

@article{CrochemoreTischler2011,
  title={Computing longest previous non-overlapping factors},
  author={Crochemore, Maxime and Tischler, German},
  journal={Information Processing Letters},
  volume={111},
  number={6},
  pages={291--295},
  year={2011},
  publisher={Elsevier}
}

@inproceedings{KarkkainenSanders2003,
  title={Simple linear work suffix array construction},
  author={K{\"a}rkk{\"a}inen, Juha and Sanders, Peter},
  booktitle={International colloquium on automata, languages, and programming},
  pages={943--955},
  year={2003},
  organization={Springer}
}

@inproceedings{KasaiEtAl2001,
  title={Linear-time longest-common-prefix computation in suffix arrays and its applications},
  author={Kasai, Toru and Lee, Gunho and Arimura, Hiroki and Arikawa, Setsuo and Park, Kunsoo},
  booktitle={Annual Symposium on Combinatorial Pattern Matching},
  pages={181--192},
  year={2001},
  organization={Springer}
}

\end{document}